\documentclass[11pt]{article}
\usepackage[left=2.50cm, right=2.50cm, top=2.50cm, bottom=2.50cm]{geometry}

\usepackage{amssymb,amsmath,amsthm,amstext,graphics,amsfonts,hyperref, xcolor, stmaryrd}
\usepackage{cite}
\usepackage{graphicx}
\usepackage{fancyhdr}
\usepackage{enumerate}
\usepackage[justification=centering]{caption}
\usepackage[labelformat=simple]{subcaption}

\usepackage[light,condensed,math]{kurier}
\usepackage[T1]{fontenc}

\usepackage{mathpazo}      
\usepackage{charter}        

\newcommand{\FF}{{\mathbb{F}}}

\newcommand{\seq}{\subseteq}
\newcommand{\ZZ}{\mathbb{Z}}

\newcommand{\CC}{\mathbb{C}}

\newcommand{\av}{\mathbf{a}}
\newcommand{\bv}{\mathbf{b}}
\newcommand{\cv}{\mathbf{c}}
\newcommand{\dv}{\mathbf{d}}
\newcommand{\rv}{\mathbf{r}}

\newcommand{\xv}{\mathbf{x}}
\newcommand{\yv}{\mathbf{y}}

\makeatother
\usepackage{etoolbox}
\apptocmd\normalsize{%
    \abovedisplayskip=9pt plus 3pt minus 6pt
    \abovedisplayshortskip=3pt plus 3pt
    \belowdisplayskip=9pt plus 3pt minus 9pt
    \belowdisplayshortskip=6pt plus 3pt minus 4pt
}{}{}

\theoremstyle{theorem}
\newtheorem{theorem}{Theorem}[section]
\newtheorem{lem}[theorem]{Lemma}
\newtheorem{cor}[theorem]{Corollary}

\theoremstyle{definition}
\newtheorem{ex}[theorem]{Example}

\title{Quantum codes constructed from cyclic codes over the ring $\FF_q+v\FF_q+v^2\FF_q+v^3\FF_q+v^4\FF_q$}
\author{Djoko Suprijanto
\footnote{Combinatorial Mathematics Research Group, Faculty of Mathematics and Natural Sciences,
Institut Teknologi Bandung,
Jl. Ganesha 10, Bandung, 40132,
INDONESIA,\hfill{djoko@math.itb.ac.id}}
\hspace{0.1cm}and Hopein Christofen Tang
\footnote{Combinatorial Mathematics Research Group,
Faculty of Mathematics and Natural Sciences,
Institut Teknologi Bandung,
Jl. Ganesha 10, Bandung, 40132,
INDONESIA, \hfill{hopeinct@students.itb.ac.id}}}
\date{}

\begin{document}

\maketitle

\begin{abstract}
In this article, we investigate properties of cyclic codes over a finite non-chain ring
$\FF_q+v\FF_q+v^2\FF_q+v^3\FF_q+v^4\FF_q,$ where $q=p^r,$ $r$ is a positive integer, $p$ is an odd prime,
$4 \mid (p-1),$ and $v^5=v.$  As an application, we construct several quantum error correcting codes over
the finite field $\FF_q.$
\end{abstract}

\section{Introduction}
Study of quantum communication and quantum computation is motivated by the fact that, in a sense, quantum mechanics is more superior to classical mechanics.  Hence, quantum bits or shorty qubits are proposed to studied, "instead of classical bits used in computers working by the rules of classical mechanics", to quote Sari and Siap \cite{SS2016a}.  Theoretically, qubits can store
more information in transition or storage compared to the classical case.  It is because of the superposition state of qubits. Beside its superiority,
"one of the main problems for qubits is the decoherence that destroys in a superposition of qubits" \cite{SS2016a}.  At first, the decoherence makes the quantum communication as well as quantum computation seems to be impossible to realize.  Fortunately, the quantum error-correcting codes (QECC) can handle this problem.

The early solution to this problem was proposed by Shor by introducing a quantum error correcting code that encoded one qubit to a highly entangled state of nine qubits \cite{Shor}.  Calderbank and Shor \cite{CS} introduced a method to construct quantum error-correcting codes from the classical ones. Their method is known later as a Calderbank-Shor-Steane construction or CSS construction, for short.   Later, Steane \cite{Steane} proposed a generalization of Calderbank-Shor-Steane construction.  His method enables
him to obtain many new quantum codes unknown to exist before.

Although the original studies of quantum error-correcting codes were over a binary field, later these studies have been generalized to non-binary fields "with the goal to relate the later codes to the former ones" (see \cite{CSS} for earliest study on it, c.f. \cite{AK}).  Later, many people constructed quantum error correcting codes from cyclic codes over various rings (see, e.g.,\cite{QMG2009}, \cite{Gao2015}, \cite{SS2016a}, \cite{SS2016b}, to mention a few of them).  Qian, Ma, and Guo \cite{QMG2009} constructed quantum error-correcting codes starting from cyclic codes over the ring $\FF_2+u\FF_2,$ with $u^2=0.$  Gao \cite{Gao2015} constructed quantum error-correcting codes from cyclic codes over the ring $\FF_q+v\FF_q+v^2\FF_q+v^3\FF_q.$  Gao's work \cite{Gao2015} can be considered as a certain generalization of the work of Sari and Siap \cite{SS2016a}, where they constructed quantum codes from cyclic codes over the ring $\FF_p+v\FF_p+\cdots+v^{p-1}\FF_p,$ where $v^p=v,$ and $p$ is a prime number.  Sari and Siap \cite{SS2016b} have also considered binary quantum error-correcting codes constructed from the cyclic codes over the ring $\FF_2+u\FF_2+u^2\FF_2+\cdots+u^{s-1}\FF_2,$
where $u^s=0.$

In this paper, continuing the investigation of Gao \cite{Gao2015}, we study the structural aspects of cyclic codes over
the ring $R:=\FF_q+v\FF_q+v^2\FF_q+v^3\FF_q+v^4\FF_q,$ where $q=p^r,$ $r$ is a positive integer, $p$ is an odd prime,
$4 \mid (p-1),$ and
$v^5=v.$  As an application, we construct quantum error-correcting codes from cyclic codes over the ring $R.$

This paper is organized as follows.  In Section 2, we consider the structure of  linear codes as well as cyclic codes over the ring $R.$  We also define a Gray map from $R^n$ to $\FF_q^{5n}$ and derive some related properties.  The construction of quantum error-correcting codes from cyclic codes $R$ is given in Section 3.  We also give several concrete examples of quantum codes over finite fields $\FF_q$ in the end of Section 3.

\section{Linear codes over $R$ and the Gray map}

Let $\FF_q$ be a finite field of order $q,$ where $q=p^r,$ $r$ is a positive integer, $p$ is an odd prime, and $4 \mid (p-1).$
Let $R$ denote the ring $\FF_q+v\FF_q+v^2\FF_q+v^3\FF_q+v^4\FF_q$ where
$v^5=v.$    It is clear that $R$ is isomorphic to the ring $\FF_q[v]/\langle v^5-v\rangle.$
It is also well known (see, e.g., Niederreiter \cite{Nieder})  that $v^{p-1}-1$ has a unique factorization into
linear factors over $\FF_q.$  Moreover, since $4 \mid (p-1),$ we have  $(v^4-1) \mid (v^{p-1}-1)$ which implies
$v^5-v=v(v^4-1)=v(v-1)(v+1)(v-a_1)(v-a_2)=v(v-1)(v+1)(v-a)(v+a),$ where
$a,a_1,a_2 \in \FF_q$ with $ a^2=-1.$

Let $f_1=v,$ $f_2=v-1,$ $f_3=v+1,$ $f_4=v-a,$ $f_5=v+a;$ and for $i \in [1,5]_{\ZZ},$ let $\widehat{f_i}=\frac{v^5-v}{f_i}.$
Then there exist $\alpha_i, \beta_i \in \FF_q[v]$ such that $\alpha_i f_i + \beta_i \widehat{f_i}=1.$
Now, for $i \in [1,5]_{\ZZ},$ let $\eta_i=\beta_i \widehat{f_i}.$ Then we have
\begin{itemize}
\item[(i)] $\eta_1,\eta_2,\eta_3,\eta_4, \eta_5 \in R$ are nonzero idempotents orthogonal in $R.$

\item[(ii)] $\eta_1+\eta_2+\eta_3+\eta_4+ \eta_5=1$ in $R.$
\end{itemize}
It implies that, by the Chinese Reminder Theorem, the ring $R$ can be decomposed as follows:
\begin{equation}\label{crt}
R=R\eta_1 \oplus R\eta_2 \oplus R\eta_3 \oplus R\eta_4 \oplus R\eta_5\\
 =\FF_q \eta_1  \oplus \FF_q \eta_2 \oplus \FF_q \eta_3 \oplus \FF_q \eta_4 \oplus \FF_q \eta_5.
\end{equation}
\begin{ex}
For $q = p = 13,$ the five roots of $(v^5-v)$ are $v_1 = 0,$ $v_2 = 1,$ $v_3 = 12=-1,$ $v_4 = 5$, and $v_5 =8=-5$ and the five idempotents are
\[
\begin{aligned}
\eta_1&=12v^4+1,\\
\eta_2&=10v^4+10v^3+10v^2+10v,\\
\eta_3&=10v^4+3v^3+10v^2+3v,\\
\eta_4&=10v^4+2v^3+3v^2+11v,\\
\eta_5&=10v^4+11v^3+3v^2+2v.
\end{aligned}
\]
\hfill $\diamondsuit$
\end{ex}


By the Equation (\ref{crt}) we know that for any $r \in R$ there exist $b_1,b_2,b_3,b_4,b_5 \in \FF_q$ such that
$r=b_1 \eta_1+b_2 \eta_2+b_3 \eta_3+b_4 \eta_4+b_5 \eta_5.$  Let us define a Gray map $\phi$ from $R$ to $\FF_q^5$ by
\[
\begin{array}{rcl}
r & \longmapsto & (b_1,b_2,b_3,b_4,b_5)
\begin{pmatrix}
6 & 2 & 3 & -6  & 6 \\
6 & 6 & 2 & 3  & -6 \\
-6 & 6 & 6 & 2  & 3 \\
3 & -6 & 6 & 6  & 2 \\
2 & 3 & -6 & 6  & 6
\end{pmatrix},
\end{array}
\]
where $6=1+1+1+1+1+1.$

Next, we define the Gray map $\Phi$ from $R^n$ to $\FF_q^{5n}$ as an extension of a Gray map $\phi$ by
\[
\begin{array}{rcl}
\rv & \longmapsto & (\phi(r_0), \phi(r_1),\ldots, \phi(r_{n-1})),
\end{array}
\]
where $\rv=(r_0,r_1,\ldots,r_{n-1}) \in R^n.$

The Hamming weight of an element $\xv \in \FF_q^{5n},$ denoted by $w_H(\xv),$ is defined as a number of nonzero components of $\xv.$
The Lee weight of the element $r=b_1 \eta_1+b_2 \eta_2+b_3 \eta_3+b_4 \eta_4+b_5 \eta_5 \in R,$ denoted by
$w_L(r),$ is defined by
\[
w_L(r)=w_H(\phi(r)).
\]

We also define the Lee weight of a vector $\rv=(r_0,r_1,\ldots,r_{n-1}) \in R^n,$ naturally, to be the rational sum
of Lee weights of its components, i.e. $w_L(\rv)=\sum_{i=0}^{n-1} w_L(r_i).$ For any vectors
$\xv_1$ and $\xv_2$ in $R^n,$ the Lee distance between $\xv_1$ and $\xv_2$ is given by $d(\xv_1,\xv_2)=w_L(\xv_1-\xv_2).$

A code $C$ of length $n$ over $R$ is defined as  a nonempty subset of $R^n.$ Any element of a code $C$ is called a codeword. A code $C$ is called linear if and only if $C$ is an
$R$-submodule of $R^n.$ The minimum Lee distance of $C$ is the smallest nonzero Lee
distance between all pairs of distinct codewords. The minimum Lee weight of $C$ is the
smallest nonzero Lee weight among all codewords. It is easy to see that if $C$ is linear, then the minimum
Lee distance of the code $C$ is the same as its minimum Lee weight. In this paper, we always assume
that $C$ is a linear code over $R.$  Note that the Hamming distance $d_H(\xv_1,\xv_2)$ is defined similarly.

Since, for every $\alpha \in \FF_q$ and $r,s \in R$ the map $\phi$ satisfies $\phi(r+\alpha s)=\phi(r)+\alpha \phi(s),$
then we have that $\Phi$ is $\FF_q$-linear.  Now, let $\xv,\yv \in R^n.$  Then $d_L(\xv,\yv)=w_L(\xv-\yv)
=w_H(\Phi(\xv-\yv))=w_H(\Phi(\xv)-\Phi(\yv))=d_H(\Phi(\xv)-\Phi(\yv)).$  Hence, we have proven the lemma below.

\begin{lem}\label{L-isometry}
The Gray map $\Phi$ is an isometry, namely a distance-preserving map, from $(R^n,~d_L)$ to
$(\FF_q^{5n},d_H).$ Moreover it is also $\FF_q$-linear.
\end{lem}

Moreover, if $C$ is a linear code over the ring $R$ then we have that the Gray image of $C$ is also
a linear code over $\FF_q.$

\begin{lem}\label{L-gmap}
Let $C$ be a $[n,A,d_L]$ linear code over $R,$ where $n,~A$ and $d$ are the code
length, the number of codewords and the minimum Lee distance of $C,$ respectively.
Then $\Phi(C)$ is a $[5n,\log_q A, d_H]$ linear code over $\FF_q.$
\end{lem}

\begin{proof}
From Lemma \ref{L-isometry}, we see that $\Phi(C)$ is a linear code over $\FF_q.$  By definition of
the Gray map $\Phi$ we have that $\Phi(C)$ is of length $5n.$  Moreover, since $\Phi$ is a bijection, which is easy
to verify, we have that $\Phi(C)$ has dimension $\log_q A.$  Finally, the isometry of $\Phi$ implies
$\Phi(C)$ has minimum Hamming distance $d_H.$
\end{proof}

Let $\xv=(x_0,x_1,\ldots,x_{n-1})$ and $\yv=(y_0,y_1,\ldots,y_{n-1})$ be two vectors in $R^n.$
Then the product of $\xv$ and $\yv$ is defined as
\[
\xv \cdot \yv=\sum_{i=0}^{n-1}x_i y_i.
\]

The dual of a code $C$ over $R,$ denoted by $C^\perp,$ is defined by
\[
C^\perp:=\{\xv\in R^n\color{black}:~\xv \cdot \yv=0,~\text{for all }\yv \in C\}.
\]
If $C \seq C^\perp,$ then $C$ is said to be a self-orthogonal code.

We also define the Euclidean inner product in $\FF_q^{n},$
\[
[\av,\bv]:=\sum_{i=0}^{n-1} a_i b_i,
\]
for all $\av,\bv \in \FF_q^{n}.$  The dual code $C$ over $\FF_q$ as well as self-orthogonality of the code
are defined similarly.

\begin{theorem}\label{T-SO}
Let $C$ be a linear code over $R.$ Then $\Phi(C^\perp)=\Phi(C)^\perp.$ Moreover, if $C$ is
self-orthogonal over $R,$ then $\Phi(C)$ is also self-orthogonal over $\FF_q.$
\end{theorem}

\begin{proof}
Let $\xv=(x_0,x_1,\ldots,x_{n-1}) \in C$ and $\yv=(y_0,y_1,\ldots,y_{n-1}) \in C^\perp.$  For $i \in [0,n-1]_\ZZ,$ let
$x_i=x_{i1}\eta_1+x_{i2}\eta_2+x_{i3}\eta_3+x_{i4}\eta_4+x_{i,5}\eta_5$ and $y_i=y_{i1}\eta_1+y_{i2}\eta_2+y_{i3}\eta_3+y_{i4}\eta_4+y_{i5}\eta_5.$
Then we have
\begin{align*}
0&=\xv\cdot\yv\\
 &=\sum_{i=0}^{n-1}x_{i}y_{i}\\
&=\sum_{i=0}^{n-1} x_{i1}y_{i1}\eta_1+x_{i2}y_{i2}\eta_2+x_{i3}y_{i3}\eta_3
+x_{i4}y_{i4}\eta_4+x_{i5}y_{i5}\eta_5.
\end{align*}
Hence, for all $j \in [1,5]_\ZZ,$ we have $\displaystyle \sum_{i=0}^{n-1}x_{ij}y_{ij}=0,$  which implies
\[
\Phi(\xv)\cdot\Phi(\yv)
=k\sum_{i=0}^{n-1}x_{i1}y_{i1}+x_{i2}y_{i2}+x_{i3}y_{i3}+x_{i4}y_{i4}+x_{i5}y_{i5}=0,
\]
with $k=6^2+6^2+(-6)^2+3^2+2^2$. Therefore, $\Phi(C^\perp)\seq\Phi(C)^\perp.$ By using the fact that $\Phi$ is a bijection, then $|\Phi(C)|=|C|$ and also
$|\Phi(C^\perp)|=|C^\perp|.$  Furthermore, since $|R^n|=|C||C^\perp|$ and $|\FF_q^{5n}|=|\Phi(C)||\Phi(C)^\perp|$, we conclude that $|\Phi(C^\perp)|=|\Phi(C)^\perp|,$ and hence $\Phi(C^\perp)=\Phi(C)^\perp.$  Moreover, if $C$ is self-orthogonal, then $\Phi(C)\seq \Phi(C^\perp)=\Phi(C)^\perp,$  and hence $\Phi(C)$ is also self-orthogonal.
\end{proof}

Now,  for all $i \in [1,5]_\ZZ,$ define a code $C_i \seq \FF_q^n$  as follows:
\[
C_1:=\{\cv_1 \in \FF_q^{n}:~\exists \cv_2,\cv_3,\cv_4,\cv_5 \in \FF_q^{n} \text{ s.t. }
\cv_1\eta_1+\cv_2\eta_2+\cv_3\eta_3+\cv_4\eta_4+\cv_5\eta_5 \in C\},
\]
\[
C_2:=\{\cv_2 \in \FF_q^{n}:~\exists \cv_1,\cv_3,\cv_4,\cv_5 \in \FF_q^{n} \text{ s.t. }
\cv_1\eta_1+\cv_2\eta_2+\cv_3\eta_3+\cv_4\eta_4+\cv_5\eta_5 \in C\},
\]
\[
C_3:=\{\cv_3 \in \FF_q^{n}:~\exists \cv_1,\cv_2,\cv_4,\cv_5 \in \FF_q^{n} \text{ s.t. }
\cv_1\eta_1+\cv_2\eta_2+\cv_3\eta_3+\cv_4\eta_4+\cv_5\eta_5 \in C\},
\]
\[
C_4:=\{\cv_4 \in \FF_q^{n}:~\exists \cv_1,\cv_2,\cv_3,\cv_5 \in \FF_q^{n} \text{ s.t. }
\cv_1\eta_1+\cv_2\eta_2+\cv_3\eta_3+\cv_4\eta_4+\cv_5\eta_5 \in C\},
\]
\[
C_5:=\{\cv_5 \in \FF_q^{n}:~\exists \cv_1,\cv_2,\cv_3,\cv_4 \in \FF_q^{n} \text{ s.t. }
\cv_1\eta_1+\cv_2\eta_2+\cv_3\eta_3+\cv_4\eta_4+\cv_5\eta_5 \in C\}.
\]

It is not difficult to see that for all $i \in [1,5]_\ZZ,$ the code $C_i$ are all linear over $\FF_q.$  Moreover,
the linear code $C$ over $R$ can be uniquely expressed as
\begin{equation}\label{oplus-C}
C=C_1 \eta_1 \oplus C_2 \eta_2 \oplus C_3 \eta_3 \oplus C_4 \eta_4 \oplus C_5 \eta_5.
\end{equation}

Let $G$ be a generator matrix of a code $C$ over $R,$
and let $\Phi(G)$ be a generator matrix of a code $\Phi(C)$ over $\FF_q.$  Then, by Equation (\ref{oplus-C}) we have
\[
G=\begin{pmatrix}
\eta_1 G_1\\ \eta_2 G_2\\ \eta_3 G_3\\ \eta_4 G_4\\ \eta_5 G_5
\end{pmatrix},
\]
and
\[
\Phi(G)=\begin{pmatrix}
\Phi(\eta_1 G_1)\\ \Phi(\eta_2 G_2)\\ \Phi(\eta_3 G_3)\\ \Phi(\eta_4 G_4)\\ \Phi(\eta_5 G_5)
\end{pmatrix},
\]
where for all $i \in [1,5]_\ZZ,$ $G_i$ is a generator matrix of $C_i.$

\begin{lem}\label{L-dual}
Let $C=C_1\eta_1 \oplus C_2 \eta_2 \oplus C_3 \eta_3 \oplus C_4 \eta_4  \oplus C_5\eta_5$
be a linear code of length $n$ over $R.$
Then
\[
C^\perp=C_1^\perp \eta_1 \oplus C_2^\perp \eta_2 \oplus C_3^\perp \eta_3 \oplus C_4^\perp \eta_4 \oplus C_5^\perp \eta_5.
\]
Moreover, $C$ is a self-orthogonal code over $R$ if and only if $C_1, C_2, C_3, C_4,$ and $C_5$ are all
self-orthogonal codes over $\FF_q.$
\end{lem}

\begin{proof}
Define
\[
\widehat{C}_1:=\{\cv_1 \in \FF_q^{n}:~\exists \cv_2,\cv_3,\cv_4,\cv_5 \in \FF_q^{n} \text{ s.t. }
\cv_1\eta_1+\cv_2\eta_2+\cv_3\eta_3+\cv_4\eta_4+\cv_5\eta_5 \in C^\perp\},
\]
\[
\widehat{C}_2:=\{\cv_2 \in \FF_q^{n}:~\exists \cv_1,\cv_3,\cv_4,\cv_5 \in \FF_q^{n} \text{ s.t. }
\cv_1\eta_1+\cv_2\eta_2+\cv_3\eta_3+\cv_4\eta_4+\cv_5\eta_5 \in C^\perp\},
\]
\[
\widehat{C}_3:=\{\cv_3 \in \FF_q^{n}:~\exists \cv_1,\cv_2,\cv_4,\cv_5 \in \FF_q^{n} \text{ s.t. }
\cv_1\eta_1+\cv_2\eta_2+\cv_3\eta_3+\cv_4\eta_4+\cv_5\eta_5 \in C^\perp\},
\]
\[
\widehat{C}_4:=\{\cv_4 \in \FF_q^{n}:~\exists \cv_1,\cv_2,\cv_3,\cv_5 \in \FF_q^{n} \text{ s.t. }
\cv_1\eta_1+\cv_2\eta_2+\cv_3\eta_3+\cv_4\eta_4+\cv_5\eta_5 \in C^\perp\},
\]
and
\[
\widehat{C}_5:=\{\cv_5 \in \FF_q^{n}:~\exists \cv_1,\cv_2,\cv_3,\cv_4 \in \FF_q^{n} \text{ s.t. }
\cv_1\eta_1+\cv_2\eta_2+\cv_3\eta_3+\cv_4\eta_4+\cv_5\eta_5 \in C^\perp\}.
\]
Then $C^\perp=\widehat{C}_1 \eta_1 \oplus \widehat{C}_2 \eta_2 \oplus \widehat{C}_3 \eta_3 \oplus \widehat{C}_4 \eta_4 \oplus \widehat{C}_5 \eta_5$ and this expression is unique.  It is easy to check that $\widehat{C}_1 \seq C_1^\perp.$
Let $\xv \in C_1^\perp.$ For any $\yv=\cv_1 \eta_1+\cv_2 \eta_2+\cv_3 \eta_3+\cv_4 \eta_4+\cv_5 \eta_5 \in C,$ we have
$\xv \eta_1 \cdot \yv=0$ which implies $\xv \eta_1 \in C^\perp.$ By the unique expression of $C^\perp,$ we have
$\xv \in \widehat{C}_1$ and hence $C_1^\perp=\widehat{C}_1.$ Similarly, for all $i=2,3,4,5$ we have $C_i^\perp=\widehat{C}_i,$
and hence we conclude that $C^\perp=C_1^\perp \eta_1 \oplus C_2^\perp \eta_2 \oplus C_3^\perp \eta_3 \oplus C_4^\perp \eta_4 \oplus C_5^\perp \eta_5.$

Moreover, it is clear that $C$ is self-orthogonal over $R$ if $C_1,C_2,C_3,C_4,$ and $C_5$ are all self-orthogonal over $\FF_q.$
Now let $C$ self-orthogonal over $R$ and $\cv=\cv_1 \eta_1 + \cv_2 \eta_2+\cv_3 \eta_3 + \cv_4 \eta_4+\cv_5 \eta_5 \in C,$ with $\cv_1 \in C_1,$ $\cv_2 \in C_2,$ $\cv_3 \in C_3,$ $\cv_4 \in C_4,$ and $\cv_5 \in C_5.$  Then for all $\dv=\dv_1 \eta_1 + \dv_2 \eta_2+\dv_3 \eta_3 + \dv_4 \eta_4+\dv_5 \eta_5 \in C,$ with $\dv_1 \in C_1,$ $\dv_2 \in C_2,$ $\dv_3 \in C_3,$ $\dv_4 \in C_4,$ and $\dv_5 \in C_5,$ we have
$0=\cv \cdot \dv=[\cv_1, \dv_1] \eta_1 + [\cv_2, \dv_2] \eta_2+[\cv_3, \dv_3] \eta_3 + [\cv_4, \dv_4] \eta_4+ [\cv_5, \dv_5] \eta_5.$
It implies $0=[\cv_1, \dv_1]=[\cv_2, \dv_2]=[\cv_3, \dv_3]=[\cv_4, \dv_4]=[\cv_5, \dv_5],$ and hence for all $i \in [1,5]_{\ZZ}$ we have $\cv_i \in C_i^\perp.$ Therefore, $C_1,$ $C_2,$ $C_3,$ $C_4,$
and $C_5$ are all self-orthogonal over $\FF_q.$
\end{proof}

\section{Quantum codes from cyclic codes over $R$}

A quantum code of length $n$ and dimension $q$ over $\FF_q$ is defined to be the subspace of the Hilbert space $\left(\CC^q \right)^{\otimes n}$ of dimension $q^k$ and we denote it by $\llbracket n,k,d \rrbracket_q.$

In the class of linear codes, cyclic codes play an important rule in coding theory.  Moreover, as we can obtain many quantum codes from cyclic codes, then we provide some related results on cyclic codes over $R.$

A linear code $C \seq R^n$ over $R$ is called cyclic if $T(C)=C.$ Here, $T$ is cyclic shift operator on $R^n,$ namely for any $\cv=(c_0,c_1,c_2,\ldots,c_{n-1}) \in R^n,$ we have $T(\cv)=(c_{n-1},c_0,c_1,\ldots,c_{n-2}).$

Define $R[x] \slash \langle x^n-1\rangle:=\{c_0+c_1x+c_2x^2+\cdots+c_{n-1}x^{n-1}+\langle x^n-1\rangle:~c_0,c_1,\ldots,c_{n-1} \in R\}.$  Now, consider the following map
\[
\lambda: R^n \longrightarrow R[x] \slash \langle x^n-1\rangle,
\]
defined by
\[
\cv=(c_0,c_1,c_2,\ldots,c_{n-1}) \longmapsto c_0+c_1x+c_2x^2+\ldots+c_{n-1}x^{n-1}
+\langle x^n-1 \rangle.
\]
For convenience, we omit the term $\langle x^n-1 \rangle$ when writing any element of $R[x] \slash \langle x^n-1\rangle$. It can be proved by ease that $\lambda$ defines an $R$-module isomorphism.  Hence, we can identify a cyclic code $C$ over $R$ as an ideal of the division ring $R[x] \slash \langle x^n-1\rangle.$

Now, we recall a celebrated method to construct quantum error-correcting codes
as introduced by Calderbank, Shor, and Steane.  This well known method is called
Calderbank-Shor-Steane construction or CSS construction (see Theorem 9 and 12 in \cite{CSS}).

\begin{theorem}\cite{CSS}(CSS construction)\label{L-CSS}
Let $C_1$ and $C_2$ be two linear codes over $\FF_q$ of parameter $[n,k_1,d_1]$ and $[n,k_2,d_2],$ respectively,
such that $C_2 \seq C_1.$  Then there exists a
quantum error-correcting code with the parameters $\llbracket n,k_1-k_2,\text{min}\{d_1,d_2^\perp\}\rrbracket_q$ where
$d_2^\perp$ denotes the minimum Hamming distance of the dual code $C_2^\perp$ of $C_2.$ Further, if
$C_2=C_1^\perp,$ then there exists a quantum error-correcting code with the parameters
$\llbracket n,2k_1-n,d_1 \rrbracket_q.$
\end{theorem}

The following two properties are easy to prove but important for our construction.

\begin{lem}\label{L-cyclic-1}
A linear code $C=C_1 \eta_1 \oplus C_2 \eta_2 \oplus C_3 \eta_3 \oplus C_4 \eta_4 \oplus C_5 \eta_5$ over $R$
is cyclic if and only if $C_1,C_2,C_3,C_4,$ and $C_5$ are all cyclic over $\FF_q.$
\end{lem}

\begin{proof}
($\Longrightarrow$)  For $i \in [1,5],$ let $(c_{i1},c_{i2},\ldots,c_{in}) \in C_i.$  Also, for $j \in [1,5],$ let $c_j=\eta_1 c_{j1}+\eta_2 c_{j2}+\eta_3 c_{j3}+\eta_4 c_{j4}+\eta_5 c_{j5}.$  Then, we have $\cv=(c_1,c_2,\ldots,c_n) \in C.$  Since $C$ is cyclic over $R,$ then it follows that
$(c_n,c_1,c_2,\ldots,c_{n-1}) \in C.$  In addition, since $(c_n,c_1,c_2,\ldots,c_{n-1})=\sum_{j=1}^5
\eta_j (c_{jn},c_{j1},c_{j2},\ldots,c_{jn-1}),$ we have that
$(c_{in},c_{i1},c_{i2},\ldots,c_{in-1}) \in C_i,$ for $i\in [1,5].$  It means $C_1,C_2,C_3,C_4,$ and $C_5$ are all cyclic codes over $\FF_q.$

($\Longleftarrow$) Suppose $C_1,C_2,C_3,C_4,$ and $C_5$ are all cyclic codes over $\FF_q.$  Let $\cv=(c_1,c_2,\ldots,c_n)\in C,$ where $c_j=\eta_1c_{j1}+\eta_2c_{j2}+\eta_3c_{j3}+\eta_4c_{j4}+\eta_5c_{j5},$ for $j \in [1,5].$  Then for $i \in [1,5],$ we have
$(c_{i1},c_{i2},\ldots,c_{in}) \in C_i,$ which implies  $(c_n,c_1,c_2,\ldots,c_{n-1})=\sum_{j=1}^5 \eta_j (c_{jn},c_{j1},c_{j2},\ldots,c_{jn-1}) \in \oplus_{j=1}^5 \eta_j C_j=C.$  Hence, $C$ is a cyclic code over $R.$
\end{proof}

\begin{lem}\label{L-cyclic-2}
Let $C=C_1 \eta_1 \oplus C_2 \eta_2 \oplus C_3 \eta_3 \oplus C_4 \eta_4 \oplus C_5 \eta_5$ be a cyclic code over $R$
of length $n.$  Then there exists a unique polynomial $g(x) \in R[x] / \langle x^n-1 \rangle$ such that
\[
C=\langle g(x) \rangle,
\]
where $g(x)=\eta_1 g_1(x)+\eta_2 g_2(x)+\eta_3 g_3(x)+\eta_4 g_4(x)+\eta_5 g_5(x)$ and
$g_1(x), g_2(x), g_3(x), g_4(x)$ and $g_5(x)$ are the generator polynomial of cyclic codes
$C_1,C_2,C_3,C_4,$ and $C_5$ over $\FF_q,$ respectively.  Moreover, $g(x)$ is a divisor of $x^n-1$ over $R.$

\end{lem}

\begin{proof}
Since $g(x)\in C_1 \eta_1 \oplus C_2 \eta_2 \oplus C_3 \eta_3 \oplus C_4 \eta_4 \oplus C_5 \eta_5=C$, then $\langle g(x) \rangle \subseteq C$. For any $\cv\in C$, there exists $\cv_i\in C_i$ for $i\in [1,5]_\ZZ$ such that $\cv=\cv_1\eta_1+\cv_2\eta_2+\cv_3\eta_3+\cv_4\eta_4+\cv_5\eta_5$. We can identify $\cv$ with the polynomial $c_1(x)\eta_1+c_2(x)\eta_2+c_3(x)\eta_3+c_4(x)\eta_4+c_5(x)\eta_5$, where $c_i(x)=a_i(x)g_i(x)\in \langle g_i(x)\rangle=C_i$ for $i\in [1,5]_\ZZ$. We have $c_1(x)\eta_1+c_2(x)\eta_2+c_3(x)\eta_3+c_4(x)\eta_4+c_5(x)\eta_5=a_1(x)g_1(x)\eta_1+a_2(x)g_2(x)\eta_2+a_3(x)g_3(x)\eta_3+a_4(x)g_4(x)\eta_4+a_5(x)g_5(x)\eta_5=a(x)g(x)\in\langle g(x) \rangle$, where $a(x)=a_1(x)\eta_1+a_2(x)\eta_2+a_3(x)\eta_3+a_4(x)\eta_4+a_5(x)\eta_5$. It implies $C \seq \langle g(x) \rangle.$ Thus, $C=\langle g(x) \rangle.$ The uniqueness of $g(x)$ follows immediately from the uniqueness of $g_1(x),$ $g_2(x),$ $g_3(x),$ $g_4(x),$ and $g_5(x).$

Since for all $i \in [1,5]_\ZZ,$ $g_i(x)$ is a divisor of $x^n-1,$ then there is $h_i(x) \in \FF_q[x]$ such that $g_i(x)h_i(x)=x^n-1.$  It follows that
$x^n-1=g(x)(\eta_1 h_1(x)+\eta_2 h_2(x)+\eta_3 h_3(x)+\eta_4 h_4(x)+\eta_5 h_5(x)).$ Hence, we conclude that $g(x)$ is a divisor of $x^n-1$  over $R.$
\end{proof}

From the Lemma \ref{L-cyclic-2}, we conclude immediately that all ideals of the ring $R[x]/\langle x^n-1\rangle$ are generated by only one element, and hence the ring is principal.  Moreover, the cardinality of the cyclic code $C$ can be easily calculated as follows:
\begin{align*}
|C|&=|C_1||C_2||C_3||C_4||C_5|\\
   &=q^{n-\deg{g_1(x)}}q^{n-\deg{g_2(x)}}q^{n-\deg{g_3(x)}}q^{n-\deg{g_4(x)}}q^{n-\deg{g_5(x)}}.
\end{align*}
Moreover, by applying Lemma \ref{L-dual} and the well-known property regarding the dual code $C^\perp$ of a cyclic code $C,$ we obtained the generator polynomial of the dual code $C^\perp.$
Hence, we deduce the following properties.

\begin{cor}\label{L-cyclic-3}
The following three properties hold.
\begin{itemize}
\item[(1)]  The quotient ring $R[x]/\langle x^n-1 \rangle$ is principal.

\item[(2)] Let $C$ be a cyclic code of length $n$ over $R$ as written in the Lemma \ref{L-cyclic-2}. Then
\[
\displaystyle |C|=q^{5n-\deg{g_1(x)}-\deg{g_2(x)}-\deg{g_3(x)}-\deg{g_4(x)}-\deg{g_5(x)}}.
\]

\item[(3)]  Let $C$ be a cyclic code of length $n$ over $R$ as written in the Lemma \ref{L-cyclic-2} and
$g_1(x)h_1(x)=g_2(x)h_2(x)=g_3(x)h_3(x)=g_4(x)h_4(x)=g_5(x)h_5(x)=x^n-1.$  Then $C^\perp=\langle h(x) \rangle,$
where $h(x)=\eta_1 h_1^\ast(x)+\eta_2 h_2^\ast(x)+\eta_3 h_3^\ast(x)+\eta_4 h_4^\ast(x)+\eta_5 h_5^\ast(x)$ and $h_i^\ast(x)$
is a reciprocal polynomial of $h_i(x),$ for $i=1,2,3,4,5.$

\end{itemize}
\end{cor}

If $C$ is a cyclic code of length $n$ over $\FF_q$ with generator polynomial $g(x)$ and $g(x)h(x)=x^n-1$, it is well-known that $C^\perp=\langle h^\ast(x)\rangle$. If $C^\perp\seq C=\langle g(x)\rangle$ then $h^\ast(x)=a(x)g(x)$ for some polynomial $a(x)$, so $x^n-1=-h^\ast(x)g^\ast(x)=-a(x)g(x)g^\ast(x)$. Conversely, if $x^n-1 \equiv 0\pmod{g(x) g^\ast(x)}$, then there exists $a(x)$ such that $-h^\ast(x)g^\ast(x)=x^n-1=a(x) g(x) g^\ast(x)$. We have $h^\ast(x)=-a(x)g(x)$, so $h^\ast(x)\in \langle g(x)\rangle$. Therefore, $C^\perp=\langle h^\ast(x)\rangle\seq \langle g(x)\rangle=C$. Hence, we have proven the property, Lemma \ref{L-cyclic-4}, that give us a necessary and sufficient condition for a cyclic code over finite fields to contains its dual.

\begin{lem}\label{L-cyclic-4}
A cyclic code $C$ over $\FF_q$ with generator polynomial $g(x)$ contains its dual if and only if
\[
x^n-1 \equiv 0\pmod{g(x) g^\ast(x)},
\]
where $g^\ast(x)$ is the reciprocal polynomial of $g(x).$
\end{lem}

Similar to the above lemma, the theorem below gives us necessary and sufficient condition for the cyclic code over $R$ to contains its dual.

\begin{theorem}\label{T-condition}
Let $C=C_1 \eta_1 \oplus C_2 \eta_2 \oplus C_3 \eta_3 \oplus C_4 \eta_4 \oplus C_5 \eta_5$ be a cyclic code
of length $n$ over $R,$ and let $C=\langle g(x) \rangle.$ Then $C^\perp \seq C$ if and only if
for all $i \in [1,5]_{\ZZ}$ we have
\[
x^n-1 \equiv 0 \pmod{g_i(x)g_i^\ast(x)}.
\]

\end{theorem}

\begin{proof}
($\Longrightarrow$)  Let $C^\perp \seq C.$  Then for any $i \in [1,5]_\ZZ,$ we have
\[
C_i^\perp \eta_i = C^\perp \eta_i\seq C \eta_i = C_i \eta_i.
\]

Therefore, for any $i \in [1,5]_\ZZ,$ we obtain $C_i^\perp \seq C_i.$  Hence, by Lemma \ref{L-cyclic-4}, we have that
for any $i \in [1,5]_\ZZ,$
\[
x^n-1 \equiv 0\pmod{g_i(x) g_i^\ast(x)}.
\]

($\Longleftarrow$)  For $i \in [1,5]_\ZZ,$ let $x^n-1 \equiv 0\pmod{g_i(x) g_i^\ast(x)}.$  Then, by Lemma \ref{L-cyclic-4}, we have
$C_i^\perp \seq C_i,$ which implies $C_i^\perp \eta_i \seq C_i \eta_i.$  Furthermore, by Lemma \ref{L-dual}, we obtain
\begin{align*}
C^\perp&=C_1^\perp \eta_1 \oplus C_2^\perp \eta_2  \oplus C_3^\perp \eta_3 \oplus C_4^\perp \eta_4  \oplus C_5^\perp \eta_5\\
& \seq C_1\eta_1 \oplus C_2 \eta_2 \oplus C_3 \eta_3 \oplus C_4 \eta_4  \oplus C_5\eta_5=C.
\end{align*}
\end{proof}

By using Lemma \ref{L-CSS} together with Theorem \ref{T-condition} and Theorem \ref{T-SO}, we have the following theorem to construct quantum
error-correcting codes directly.

\begin{theorem}
Let $C=C_1 \eta_1 \oplus C_2 \eta_2 \oplus C_3 \eta_3 \oplus C_4 \eta_4 \oplus C_5 \eta_5$ be a cyclic code
of length $n$ over $R$ and let $\Phi(C)$ be a linear code of parameters $[5n,k,d_L]$ over $\FF_q,$ where $d_L$ is the
minimum Lee distance of $C.$  If $C^\perp \seq C,$ then there exists a quantum error-correcting code of parameters
$\llbracket 5n,2k-5n,d_L \rrbracket$ over $\FF_q.$
\end{theorem}

Let us look at several concrete examples.

\subsection{Examples}
We provide several examples of constructing quantum error-correcting codes over finite fields $\FF_5,$ $\FF_{13},$  $\FF_{17},$ and $\FF_{29}.$

\begin{ex}
    Let $R=\FF_5+v\FF_5+v^2\FF_5+v^3\FF_5+v^4\FF_5$ and
    $n=5.$ We have  $x^5-1= (x+4)^5$ over $\FF_5.$  Let  $g(x)=\eta_1 g_1(x)+\eta_2 g_2(x)+\eta_3 g_3(x)+\eta_4 g_4(x)+\eta_5 g_5(x)$  with $g_1(x)=(x+4)^2, g_2(x)=g_3(x)=g_4(x)=g_5(x)=x+4,$ and let  $C=\langle g(x) \rangle$  be a cyclic code over $R.$  Then $\Phi(C)$ is a linear code with parameters $[25,19,3].$ Since $x^n-1 \equiv 0 \pmod{g_i(x)g_i^\ast(x)}$ for all $i \in [1,5]_{\ZZ},$  then $C^\perp\seq C$ and $\Phi(C)^\perp\seq \Phi(C).$ Therefore, there exists a quantum error-correcting code of parameters $\llbracket 25,13,3\rrbracket_5.$
\hfill $\diamondsuit$
\end{ex}

\begin{ex}
	Let $R=\FF_5+v\FF_5+v^2\FF_5+v^3\FF_5+v^4\FF_5$ and
	$n=8.$ We have  $x^8-1= (x+1)(x+2)(x+3)(x+4)(x^2+2)(x^2+3)$ over $\FF_5.$  Let  $g(x)=\eta_1 g_1(x)+\eta_2 g_2(x)+\eta_3 g_3(x)+\eta_4 g_4(x)+\eta_5 g_5(x)$  with $g_1(x)=g_2(x)=x+2, g_3(x)=x+3, g_4(x)=x^2+2, g_5(x)=x^2+3,$ and let  $C=\langle g(x) \rangle$  be a cyclic code over $R$.  Then $\Phi(C)$ is a linear code with parameters $[40,33,4].$ Since $x^n-1 \equiv 0 \pmod{g_i(x)g_i^\ast(x)}$ for all $i \in [1,5]_{\ZZ},$  then $C^\perp\seq C$ and $\Phi(C)^\perp\seq \Phi(C).$ Therefore, there exists a quantum error-correcting code of parameters $\llbracket 40,26,4 \rrbracket_5.$ \hfill $\diamondsuit$
\end{ex}

\begin{ex}
    Let $R=\FF_5+v\FF_5+v^2\FF_5+v^3\FF_5+v^4\FF_5$ and
    $n=10.$ We have  $x^{10}-1= (x+1)^5(x+4)^5$ over $\FF_5.$  Let  $g(x)=\eta_1 g_1(x)+\eta_2 g_2(x)+\eta_3 g_3(x)+\eta_4 g_4(x)+\eta_5 g_5(x)$  with $g_1(x)=g_2(x)=x+1, g_3(x)=g_4(x)=x+4, g_5(x)=(x+4)^2,$ and let  $C=\langle g(x) \rangle$  be a cyclic code over $R.$  Then $\Phi(C)$ is a linear code with parameters $[50,44,3].$ Since $x^n-1 \equiv 0 \pmod{g_i(x)g_i^\ast(x)}$ for all $i \in [1,5]_{\ZZ},$  then $C^\perp\seq C$ and $\Phi(C)^\perp\seq \Phi(C).$ Therefore, there exists a quantum error-correcting code of parameters $\llbracket 50,38,3 \rrbracket_5.$ \hfill $\diamondsuit$
\end{ex}

\begin{ex}
	Let $R=\FF_{13}+v\FF_{13}+v^2\FF_{13}+v^3\FF_{13}+v^4\FF_{13}$ and $n=3.$ We have  $x^3-1= (x+4)(x+10)(x+12)$ over $\FF_{13}.$  Let  $g(x)=\eta_1 g_1(x)+\eta_2 g_2(x)+\eta_3 g_3(x)+\eta_4 g_4(x)+\eta_5 g_5(x)$  with $g_1(x)=g_2(x)=g_3(x)=x+4$, $g_4(x)=g_5(x)=x+10,$ and let  $C=\langle g(x) \rangle$  be a cyclic code over $R.$  Then $\Phi(C)$ is a linear code with parameters $[15,10,3].$ Since $x^n-1 \equiv 0 \pmod{g_i(x)g_i^\ast(x)}$ for all $i \in [1,5]_{\ZZ},$  then $C^\perp\seq C$ and $\Phi(C)^\perp\seq \Phi(C).$ Therefore, there exists a quantum error-correcting code of parameters $\llbracket 15,5,3 \rrbracket_{13}.$ \hfill $\diamondsuit$
\end{ex}

\begin{ex}
	Let $R=\FF_{13}+v\FF_{13}+v^2\FF_{13}+v^3\FF_{13}+v^4\FF_{13}$ and $n=4.$ We have  $x^4-1= (x+1)(x+5)(x+8)(x+12)$ over $\FF_{13}.$  Let  $g(x)=\eta_1 g_1(x)+\eta_2 g_2(x)+\eta_3 g_3(x)+\eta_4 g_4(x)+\eta_5 g_5(x)$  with $g_1(x)=g_2(x)=g_3(x)=g_4(x)=g_5(x)=x+8,$ and let  $C=\langle g(x) \rangle$  be a cyclic code over $R.$  Then $\Phi(C)$ is a linear code with parameters $[20,15,2].$ Since $x^n-1 \equiv 0 \pmod{g_i(x)g_i^\ast(x)}$ for all $i \in [1,5]_{\ZZ},$  then $C^\perp\seq C$ and $\Phi(C)^\perp\seq \Phi(C).$ Therefore, there exists a quantum error-correcting code of parameters $\llbracket 20,10,2 \rrbracket_{13}.$ \hfill $\diamondsuit$
\end{ex}

\begin{ex}
	Let $R=\FF_{13}+v\FF_{13}+v^2\FF_{13}+v^3\FF_{13}+v^4\FF_{13}$ and $n=6.$ We have  $ x^6-1= (x+1)(x+3)(x+4)(x+9)(x+10)(x+12)$ over $\FF_{13}.$  Let  $g(x)=\eta_1 g_1(x)+\eta_2 g_2(x)+\eta_3 g_3(x)+\eta_4 g_4(x)+\eta_5 g_5(x)$  with $g_1(x)=g_2(x)=x+3, g_3(x)=x+4, g_4(x)=x+9, g_5(x)=x+10,$ and let  $C=\langle g(x) \rangle$  be a cyclic code over $R.$  Then $\Phi(C)$ is a linear code with parameters $[30,25,3]$. Since $x^n-1 \equiv 0 \pmod{g_i(x)g_i^\ast(x)}$ for all $i \in [1,5]_{\ZZ},$  then $C^\perp\seq C$ and $\Phi(C)^\perp\seq \Phi(C).$ Therefore, there exists a quantum error-correcting code of parameters $\llbracket 30,20,3 \rrbracket_{13}.$ \hfill $\diamondsuit$
\end{ex}

\begin{ex}
	Let $R=\FF_{17}+v\FF_{17}+v^2\FF_{17}+v^3\FF_{17}+v^4\FF_{17}$ and $n=4.$ We have  $x^4-1= (x+1)(x+4)(x+13)(x+16)$ over $\FF_{17}.$  Let  $g(x)=\eta_1 g_1(x)+\eta_2 g_2(x)+\eta_3 g_3(x)+\eta_4 g_4(x)+\eta_5 g_5(x)$  with $g_1(x)=x+4, g_2(x)=g_3(x)=g_4(x)=g_5(x)=x+13,$ and let  $C=\langle g(x) \rangle$  be a cyclic code over $R.$  Then $\Phi(C)$ is a linear code with parameters $[20,15,2].$ Since $x^n-1 \equiv 0 \pmod{g_i(x)g_i^\ast(x)}$ for all $i \in [1,5]_{\ZZ},$  then $C^\perp\seq C$ and $\Phi(C)^\perp\seq \Phi(C).$ Therefore, there exists a quantum error-correcting code of parameters $\llbracket 20,10,2 \rrbracket_{17}.$ \hfill $\diamondsuit$
\end{ex}

\begin{ex}
	Let $R=\FF_{17}+v\FF_{17}+v^2\FF_{17}+v^3\FF_{17}+v^4\FF_{17}$ and $n=12.$ We have  $x^{12}-1= (x+1)(x+4)(x+13)(x+16)(x^2+x+1)(x^2+4x+16)(x^2+13x+16)(x^2+16x+1)$ over $\FF_{17}.$  Let  $g(x)=\eta_1 g_1(x)+\eta_2 g_2(x)+\eta_3 g_3(x)+\eta_4 g_4(x)+\eta_5 g_5(x)$  with $g_1(x)=x+4, g_2(x)=x^2+4x+16, g_3(x)=(x+13)(x^2+13x+16), g_4(x)=(x+4)(x^2+16x+1), g_5(x)=(x+4)(x^2+13x+16),$ and let  $C=\langle g(x) \rangle$  be a cyclic code over $R.$  Then $\Phi(C)$ is a linear code with parameters $[60,48,2].$ Since $x^n-1 \equiv 0 \pmod{g_i(x)g_i^\ast(x)}$ for all $i \in [1,5]_{\ZZ},$  then $C^\perp\seq C$ and $\Phi(C)^\perp\seq \Phi(C).$ Therefore, there exists a quantum error-correcting code of parameters $\llbracket 60,36,2 \rrbracket_{17}.$ \hfill $\diamondsuit$
\end{ex}

\begin{ex}
	Let $R=\FF_{29}+v\FF_{29}+v^2\FF_{29}+v^3\FF_{29}+v^4\FF_{29}$ and $n=7.$ We have  $x^7-1= (x+4)(x+5)(x+6)(x+9)(x+13)(x+22)(x+28)$ over $\FF_{29}.$  Let  $g(x)=\eta_1 g_1(x)+\eta_2 g_2(x)+\eta_3 g_3(x)+\eta_4 g_4(x)+\eta_5 g_5(x)$  with $g_1(x)=x+4, g_2(x)=x+5, g_3(x)=x+9, g_4(x)=x+13, g_5(x)=(x+4)(x+5),$ and let  $C=\langle g(x) \rangle$  be a cyclic code over $R.$  Then $\Phi(C)$ is a linear code with parameters $[35,29,5].$ Since $x^n-1 \equiv 0 \pmod{g_i(x)g_i^\ast(x)}$ for all $i \in [1,5]_{\ZZ},$  then $C^\perp\seq C$ and $\Phi(C)^\perp\seq \Phi(C).$ Therefore, there exists a quantum error-correcting code of parameters $\llbracket 35,23,5 \rrbracket_{29}.$ \hfill $\diamondsuit$
\end{ex}

\begin{ex}
	Let $R=\FF_{29}+v\FF_{29}+v^2\FF_{29}+v^3\FF_{29}+v^4\FF_{29}$ and $n=8.$ We have  $x^8-1= (x+1)(x+12)(x+17)(x+28)(x^2+12)(x^2+17)$ over $\FF_{29}.$  Let  $g(x)=\eta_1 g_1(x)+\eta_2 g_2(x)+\eta_3 g_3(x)+\eta_4 g_4(x)+\eta_5 g_5(x)$  with $g_1(x)=g_2(x)=g_3(x)=g_4(x)=x+17, g_5(x)=x^2+17,$ and let  $C=\langle g(x) \rangle$  be a cyclic code over $R.$  Then $\Phi(C)$ is a linear code with parameters $[40,34,3].$ Since $x^n-1 \equiv 0 \pmod{g_i(x)g_i^\ast(x)}$ for all $i \in [1,5]_{\ZZ},$  then $C^\perp\seq C$ and $\Phi(C)^\perp\seq \Phi(C).$ Therefore, there exists a quantum error-correcting code of parameters $\llbracket 40,28,3 \rrbracket_{29}.$  \hfill $\diamondsuit$
\end{ex}

\section{Concluding remarks}
We investigated the structures of cyclic codes over the finite non-chain ring
$\FF_q+v\FF_q+v^2\FF_q+v^3\FF_q+v^4\FF_q$ where $v^5=v.$  Several properties of cyclic codes are derived.  As an application, we construct several quantum error-correcting codes with certain parameters.  It is very interesting to investigate the structures of cyclic codes over the more general non-chain ring as mention above, namely the ring $\FF_q+v\FF_q+v^2\FF_q+v^3\FF_q+\cdots+v^{m-1}\FF_q$ where
$v^m=v$ as well as their application in constructing the quantum codes.

%

\end{document}